\documentclass[a4paper, 10pt, conference]{ieeeconf}      

\IEEEoverridecommandlockouts                              
\usepackage{url}

\usepackage{xspace}
\usepackage[english]{babel}
\usepackage[latin1]{inputenc}

\usepackage[svgnames]{xcolor}
\usepackage{gastex}

\usepackage{amsmath,amssymb}
\usepackage{enumerate}
 
\usepackage[colorlinks=true,citecolor=blue]{hyperref}

\newcommand{\stavros}[1]{#1}

\newcommand{\ST}[1]{\mbox{}}
\newcommand{\FC}[1]{\mbox{}}
\def\ie{{i.e.},~}

\def\st{{s.t.}~}
\def\Trace{\textit{Tr}}
\def\trace{\textit{tr}}
\def\faulty{\textit{Faulty}}
\def\nonfaulty{\textit{NonFaulty}}
\def\runs{\textit{Runs}}
\def\lang{{\cal L}}

\def\cost{\textit{Cost}} 

\newcommand{\proj}[1]{\boldsymbol{\pi}_{/#1}} 

\newcommand{\setN}{\mathbb N}

\newcommand{\setZ}{\mathbb Z}
\newcommand{\setQ}{\mathbb Q}

\def\calA{{\cal A}}
\def\calB{{\cal B}}

\def\calF{{\cal F}}
\def\calO{{\cal O}}

\def\enabled{\textit{en}}

\newcommand{\out}[1]{{\mathit{Out}}(#1)}

\def\last{\textit{tgt}}

\def\obs{\mbox{\normalfont Obs}\xspace}

\usepackage{theorem}
\newtheorem{definition}{Definition}
\newtheorem{example}{Example}
\newtheorem{theorem}{Theorem}
\newtheorem{lemma}{Lemma}
\newtheorem{prob}{Problem}

\newtheorem{remark}{Remark}

\def\endef{\ifmmode\squareforged\else{\unskip\nobreak\hfil
\penalty50\hskip1em\null\nobreak\hfil$\blacksquare$
\parfillskip=0pt\finalhyphendemerits=0\endgraf}\fi}

\def\qed{\ifmmode\squareforged\else{\unskip\nobreak\hfil
\penalty50\hskip1em\null\nobreak\hfil$\Box$
\parfillskip=0pt\finalhyphendemerits=0\endgraf}\fi}

 \title{\LARGE \bf Fault Diagnosis with Dynamic
  Observers$^\ast$\thanks{$^\ast$Preliminary versions of parts of this paper
    appeared in~\cite{cassez-tase-07} and \cite{cassez-acsd-07}.}}

\author{
Franck Cassez$^\dagger$\thanks{$^\dagger$
  Work suported by the French government under grant ANR-06-SETI.  }
\\
CNRS, IRCCyN Laboratory\\ 1 rue de la Noë \\ BP 92101 \\ 44321 Nantes
Cedex 3 \\ France \\ Email: franck.cassez@cnrs.irccyn.fr.  \and
Stavros Tripakis
\\
Cadence Research Laboratories \\ 2150 Shattuck Avenue, 10th floor \\
Berkeley, CA, 94704 \\ USA \\
and \\ CNRS, Verimag Laboratory \\ Centre Equation \\ 2, avenue de Vignate,
38610 Gi\`eres \\ France \\ Email: tripakis@cadence.com.
}

\begin{document}

\maketitle
\thispagestyle{empty}
\pagestyle{plain}



\begin{abstract}
  In this paper, we review some recent results about the use of
  dynamic observers for fault diagnosis of discrete event systems.
  Fault diagnosis consists in synthesizing a diagnoser that observes a
  given plant and identifies faults in the plant as soon as possible
  after their occurrence. Existing literature on this problem has
  considered the case of fixed static observers, where the set of
  observable events is fixed and does not change during execution of
  the system.  In this paper, we consider dynamic observers: an
  observer can ``switch'' sensors on or off, thus dynamically changing
  the set of events it wishes to observe.  It is known that checking
  diagnosability (\ie whether a given observer is capable of
  identifying faults) can be solved in polynomial time for static
  observers, and we show that the same is true for dynamic ones.  We
  also solve the problem of dynamic observers' synthesis and prove
  that a most permissive observer can be computed in doubly
  exponential time, using a game-theoretic approach.  We further
  investigate optimization problems for dynamic observers and define a
  notion of cost of an observer.
\end{abstract}

\section{Introduction}
\label{sec_intro}
\subsection{Monitoring, Testing, Fault Diagnosis and Control}
Many problems concerning the monitoring, testing, fault diagnosis and
control of discrete event systems (DES) can be formalized  using
finite automata over a set of {\em observable} events $\Sigma$, plus a
set of {\em unobservable}
events~\cite{RamadgeWonham87,Tsitsiklis89}. The invisible actions can
often be represented by a single unobservable event
$\varepsilon$. Given a finite automaton over $\Sigma \cup
\{\varepsilon\}$ which is a model of a {\em plant} (to be monitored,
tested, diagnosed or controlled) and an {\em objective} (good
behaviours, what to test for, faulty behaviours, control objective) we
want to check if a monitor/tester/diagnoser/controller exists that
achieves the objective, and if possible to synthesize one
automatically.

The usual assumption in this setting is that the set of observable
events is fixed (and this in turn, determines the set of unobservable
events as well).  Observing an event usually requires some detection
mechanism, \ie a \emph{sensor} of some sort. Which sensors to use,
how many of them, and where to place them are some of the design
questions that are often difficult to answer, especially without
knowing what these sensors are to be used for.

In this paper we review some recent results about \emph{sensor
  minimization}. These results are interesting since observing an
event can be costly in terms of time or energy: computation time must
be spent to read and process the information provided by the sensor,
and power is required to operate the sensor (as well as perform the
computations).  It is then essential that the sensors used really
provide useful information.  It is also important for the computer to
discard any information given by a sensor that is not really needed.

In the case of a fixed set of observable events, it is not the case
that all sensors \emph{always} provide useful information and
sometimes energy (used for sensor operation and computer treatment) is
spent for nothing. For example, to detect a fault $f$ in the system
described by the automaton $\calB$, Figure~\ref{fig-ex3},
page~\pageref{fig-ex3}, an observer needs to watch only for event $a$
initially, and watch for event $b$ \emph{only after $a$ has
  occurred}. If the sequence $a.b$ occurs, for sure $f$ has occurred
and the observer can raise an alarm.  If, on the other hand,
  event $b$ is not observed after $a$, then $f$ has not occurred. It
  is then not useful to switch on sensor $b$ before observing event
  $a$.

\subsection{Sensor Minimization and Fault Diagnosis}
We focus our attention on sensor minimization, without looking at
problems related to sensor placement, choosing between different types
of sensors, and so on.  We also focus on a particular observation
problem, that of {\em fault diagnosis}. We believe, however, that the
results we obtain are applicable to other contexts as well.

Fault diagnosis consists in observing a plant and detecting whether a
fault has occurred or not. We follow the discrete-event system (DES)
setting of~\cite{Raja95} where the behavior of the plant is known and
a model of it is available as a finite-state automaton over $\Sigma
\cup \{\varepsilon,f\}$ where $\Sigma$ is the set of potentially
observable events, $\varepsilon$ represents the unobservable events,
and $f$ is a special unobservable event that corresponds to the
faults\footnote{Different types of faults could also be considered, by
  having different fault events $f_1,f_2,$ and so on. Our methods can
  be extended in a straightforward way to deal with multiple
  faults. We restrict our presentation to a single fault event for the
  sake of simplicity.  }. Checking {\em diagnosability} (whether a
fault can be detected) for a given plant and a {\em fixed} set of
observable events can be done in polynomial
time~\cite{Raja95,Yoo-01,Jiang-01}.  In the general case, synthesizing
a diagnoser involves determinization and thus cannot be done in
polynomial time.

In this paper, we focus on {\em dynamic} observers.  For results about
sensor optimizition with \emph{static} observers, we refer the reader
to~\cite{cassez-acsd-07}.



In the dynamic observers' framework, we assume that an observer can
decide after each new observation the set of events it is going to
watch.  We first prove that checking diagnosability with dynamic
observers that are given by finite automata can be done in polynomial
time.  As a second aspect, we focus on the \emph{dynamic observer
  synthesis problem}. We show that computing a \emph{dynamic observer}
for a given plant, can be reduced to a \emph{game problem}. We further
investigate optimization problems for dynamic observers and define a
notion of \emph{cost} of an observer. Finally we show how to compute
an optimal (cost-wise) dynamic observer.

\subsection{Related Work}
To our knowledge, the problems of synthesizing dynamic observers for
diagnosability, studied in Section~\ref{sec-dyn-obs}, have not been
addressed previously in the literature.  Consequently, the associated
optimization problems, addressed in section~\ref{sec-opt-pb}, of
computing an optimal observer is also original and new.

\subsection{Organisation of the paper.}
In Section~\ref{sec-prelim} we fix notation and introduce finite
automata with faults to model DES. 

In Section~\ref{sec-dyn-obs} we introduce and study dynamic observers
and show that the most permissive dynamic observer can be computed as
the strategy in a safety 2-player game.

We also define a notion of cost for dynamic observers in
Section~\ref{sec-opt-pb} and show that the cost of a given observer
can be computed using Karp's algorithm.  Finally, we define the
optimal-cost observer synthesis problem and show it can be solved
using Zwick and Paterson's result on graph games.

This paper contains no proofs and the interested reader may refer to
\cite{cassez-tase-07,cassez-acsd-07,sensor-rr-07} for the details.

\section{Preliminaries}
\label{sec-prelim}
\subsection{Words and Languages}

Let $\Sigma$ be a finite alphabet and $\Sigma^\varepsilon= \Sigma \cup
\{\varepsilon\}$. $\Sigma^*$ is the set of finite words over $\Sigma$
and contains $\varepsilon$ which is also the empty word and
$\Sigma^+=\Sigma^* \setminus \{\varepsilon\}$. A \emph{language} $L$
is any subset of $\Sigma^*$.  Given two words $\rho,\rho'$ we denote
$\rho.\rho'$ the concatenation of $\rho$ and $\rho'$ which is defined
in the usual way.  $|\rho|$ stands for the length of the word $\rho$
\stavros{(the length of the empty word is zero)}
and ${|\rho|}_\lambda$ with $\lambda \in \Sigma$ stands for the number
of occurrences of $\lambda$ in $\rho$.
\stavros{We also use the notation $|S|$ to denote the cardinality of a set $S$.}
Given $\Sigma_1 \subseteq \Sigma$, we define the \emph{projection}
\stavros{operator on words,}
$\proj{\Sigma_1} : \Sigma^* \rightarrow \Sigma_1^*$, recursively
as follows:
$\proj{\Sigma_1}(\varepsilon)=\varepsilon$ and for $a \in \Sigma, \rho
\in \Sigma^*$, $\proj{\Sigma_1}(a.\rho)=a.\proj{\Sigma_1}(\rho)$ if $a
\in \Sigma_1$ and $\proj{\Sigma_1}(\rho)$ otherwise.

\subsection{Finite Automata}
%
\begin{definition}[Finite Automaton]\label{def-fa}
  An \emph{automaton} $A$ is a tuple
  $(Q,q_0,\Sigma^{\varepsilon},\delta)$ with $Q$ a set of
  states\footnote{In this paper we often use finite automata that
    generate prefix-closed languages, hence we do not need to use a
    set of final or accepting states.}, $q_0 \in Q$ is the initial
  state, $\delta \subseteq Q \times \Sigma^{\varepsilon} \times 2^Q$
  is the transition relation.  We write $q \xrightarrow{\, \lambda \,}
  q'$ if $q' \in \delta(q,\lambda)$.  For $q \in Q$, $\enabled(q)$ is
  the set of actions enabled at $q$.

  \noindent If $Q$ is finite, $A$ is a \emph{finite automaton}.  An
  automaton is \emph{deterministic} if for any $q \in Q$,
  $|\delta(q,\varepsilon)|=0$ and for any $\lambda \neq \varepsilon$,
  $|\delta(q,\lambda)| \leq 1$.  A \emph{labeled} automaton $A$ is a
  tuple $(Q,q_0,\Sigma,\delta,L)$ where $(Q,q_0,\Sigma,\delta)$ is an
  automaton and $L : Q \rightarrow P$ where $P$ is a finite set of
  \stavros{{\em observations}}. \endef
\end{definition}

A \emph{run} $\rho$ from state $s$ in $A$ is a \stavros{finite or infinite}
sequence of transitions
\[s_0 \xrightarrow{\, \lambda_1 \,} s_1 \xrightarrow{\, \lambda_2 \,}
s_2 \cdots s_{n-1} \xrightarrow{\, \lambda_{n} \,} s_{n} \cdots \] \st
$\lambda_i \in \Sigma^{\varepsilon}$ and $s_0=s$. If $\rho$ is finite
and ends in $s_n$ we let $\last(\rho)=s_n$. The set of finite runs
from $s$ in $A$ is denoted $\runs(s,A)$ and we define
$\runs(A)=\runs(q_0,A)$. The \emph{trace} of the run $\rho$,
de\-no\-ted $\trace(\rho)$, is the word obtained by concatenating the
symbols $\lambda_i$ appearing in $\rho$, for those $\lambda_i$
different from $\varepsilon$.  A word $w$ is \emph{accepted} by $A$ if
$w=\trace(\rho)$ for some $\rho \in \runs(A)$.  The \emph{language}
$\lang(A)$ of $A$ is the set of words accepted by $A$.

\medskip Let $f \not\in \Sigma^\varepsilon$ be a fresh letter that
corresponds to the fault action,
$\Sigma^{\varepsilon,f}=\Sigma^{\varepsilon}\cup \{f\}$ and
$A=(Q,q_0,\Sigma^{\varepsilon,f},\delta)$.  Given  $R \subseteq
\runs(A)$, $\Trace(R)=\{\trace(\rho) \textit{ for $\rho \in R$}\}$ is
the set of traces of the runs in $R$. A run $\rho$ is
\emph{$k$-faulty} if there is some $1 \leq i \leq n$ \st $\lambda_i=f$
and $n-i \geq k$. \stavros{Notice that $\rho$ can be either finite or infinite:
if it is infinite, $n=\infty$ and $n-i\ge k$ always holds.}
 $\faulty_{\geq k}(A)$ is the set of $k$-faulty runs
of $A$. A run is \emph{faulty} if it is $k$-faulty for some $k \in
\setN$ and $\faulty(A)$ denotes the set of faulty runs. It follows
that $\faulty_{\geq k+1}(A) \subseteq \faulty_{\geq k}(A) \subseteq
\cdots \subseteq \faulty_{\geq 0}(A) = \faulty(A)$.  Finally,
$\nonfaulty(A)=\runs(A) \setminus \faulty(A)$ is the set on
\emph{non-faulty} runs of $A$. We let $\faulty^{\textit{tr}}_{\geq
  k}(A)=\Trace(\faulty_{\geq k}(A))$ and
$\nonfaulty^{\textit{tr}}(A)=\Trace(\nonfaulty(A))$ be the sets of
traces of faulty and non-faulty runs.

We assume that each faulty run of $A$ of length $n$ can be extended
into a run of length $n+1$. This is required for technical reasons (in
order to guarantee that the set of faulty runs where sufficient time
has elapsed after the fault is well-defined) and can be achieved by
adding $\varepsilon$ loop-transitions to each deadlock state of $A$.
Notice that this transformation does not change the observations
produced by the plant, thus, any observer synthesized for the
transformed plant also applies to the original one.

\iftrue
\subsection{Product of Automata}
\label{sec-synchro}
\stavros{The product of automata with $\varepsilon$-transitions is defined
in the usual way: the automata} synchronize on common labels except for
$\varepsilon$.  Let
$A_1=(Q_1,q_0^1,\Sigma_1^\varepsilon,\rightarrow_1)$ and
$A_2=(Q_2,q_0^2,\Sigma_2^\varepsilon,$ $\rightarrow_2)$. The
\emph{product} of $A_1$ and $A_2$ is the automaton $A_1 \times
A_2=(Q,q_0,\Sigma,\rightarrow)$ where:
\begin{itemize}
\item $Q=Q_1 \times Q_2$,
\item $q_0=(q_0^1,q_0^2)$,
\item $\Sigma=\Sigma_1 \cup \Sigma_2$,
\item $\rightarrow \subseteq Q \times \Sigma \times Q$ is defined by
  $(q_1,q_2) \xrightarrow{\sigma} (q'_1,q'_2)$ if:
  \begin{itemize}
  \item either $\sigma \in \Sigma_1 \cap \Sigma_2$ and $q_k
    \xrightarrow{\sigma}_k q'_k$, for $k=1,2$,
  \item or $\sigma \in (\Sigma_i \setminus \Sigma_{3-i}) \cup
    \{\varepsilon\}$ and $q_i \xrightarrow{\sigma}_i q'_i$ and
    $q'_{3-i}=q_{3-i}$, \stavros{for $i=1$ or $i=2$}.
  \end{itemize}
\end{itemize}
\fi

\section{Fault Diagnosis with Dynamic Observers}
\label{sec-dyn-obs}


In this section we introduce \emph{dynamic observers}.  They can
choose after each new observation the set of events they are going to
\stavros{watch for}. To illustrate why dynamic observers can be useful
consider the following example.
\begin{example}[Dynamic Observation]
  Assume we want to detect faults in automaton $\calB$ of
  Figure~\ref{fig-ex3}.  A static diagnoser that observes
  $\Sigma=\{a,b\}$ can detect faults. However, no proper subset of
  $\Sigma$ can be used to detect faults in $\calB$. Thus the minimum
  cardinality of the set of observable events for diagnosing $\calB$
  is $2$ \ie a static observer will have to monitor two events during
  the execution of the DES.
  This means that an observer will have to be receptive to at least
  two inputs at each point in time to detect a fault in $\calB$. One
  can think of being receptive as switching on a device to sense an
  event. This consumes energy.  We can be more efficient using a
  dynamic observer, that only turns on sensors when needed, thus
  saving energy. In the case of $\calB$, this can be done as follows:
  in the beginning we only switch on the $a$-sensor; once an $a$
  occurs the $a$-sensor is switched off and the $b$-sensor is switched
  on.  Compared to the previous diagnosers we use half as much energy.
    \begin{figure}[thbtp]
    \centering
    \begin{picture}(70,16)(4,4)
      \gasset{Nframe=n,Nw=2,Nh=2,loopdiam=5,loopangle=0}
      \node[Nmarks=i](a)(10,10){$\bullet$} \node(b)(30,15){$\bullet$}
      \node(c)(50,15){$\bullet$} \node(d)(70,15){$\bullet$}
      \node(dd)(50,5){$\bullet$} \node(e)(30,5){$\bullet$}
      \drawloop(dd){$\varepsilon$}
      \drawloop(d){$\varepsilon$}
      \drawedge(a,b){$f$}
      \drawedge(b,c){$a$}
      \drawedge(c,d){$b$}
      \drawedge[ELside=r,curvedepth=0](a,e){$b$}
      \drawedge[curvedepth=0,ELside=l](e,dd){$a$}
    \end{picture}
    \caption{The automaton $\cal B$}
    \label{fig-ex3}
  \end{figure}
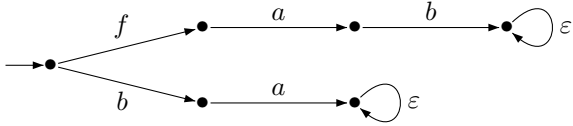
\end{example}

\subsection{Dynamic Observers}
We formalize the above notion of dynamic observation using {\em
  observers}.  The choice of the events to observe can depend on the
choices the observer has made before and on the observations it has
made. Moreover an observer may have \emph{unbounded} memory.

\begin{definition}[Observer]\label{def-observer2}
  An \emph{observer} \obs over $\Sigma$ is a deterministic labeled
  automaton $\obs=(S,s_0,\Sigma,\delta,L)$, where $S$ is a (possibly
  infinite) set of states, $s_0\in S$ is the initial state, $\Sigma$
  is the set of observable events, $\delta : S \times \Sigma
  \rightarrow S$ is the transition function (a total function), and $L
  : S \rightarrow 2^\Sigma$ is a labeling function that specifies the
  set of events that the observer wishes to observe when it is at
  state $s$.  We require for any state $s$ and any $a\in\Sigma$, if
  $a\not\in L(s)$ then $\delta(s,a)=s$: this means the observer does
  not change its state when an event it has chosen not to observe
  occurs.   \endef
\end{definition}
As an observer is deterministic we use the notation $\delta(s_0,w)$ to
denote the state $s$ reached after reading the word $w$ and
$L(\delta(s_0,w))$ is the set of events $\obs$ observes after $w$.
\medskip

\noindent An observer implicitly defines a {\em transducer} that
consumes an input event $a\in\Sigma$ and, depending on the current
state $s$, either outputs $a$ (when $a\in L(s)$) and moves to a new
state $\delta(s,a)$, or outputs $\varepsilon$, (when $a\not\in L(s)$)
and remains in the same state waiting for a new event. Thus, an
observer defines a mapping \obs from $\Sigma^*$ to $\Sigma^*$ (we use
the same name ``Obs'' for the automaton and the mapping). Given a run
$\rho$, $\obs(\proj{\Sigma}(\trace(\rho)))$ is the output of the
transducer on $\rho$.  It is called the \emph{observation} of $\rho$
by Obs.  We next provide an example of a particular case of observer
which can be represented by a finite-state machine.

\begin{example}
\begin{figure}[hbtp]
    \centering
    \begin{picture}(50,16)(10,-4)
       \gasset{Nframe=y,Nadjust=wh,Nadjustdist=2,loopdiam=5,loopangle=0}
      \node[Nmarks=i](a)(10,5){$0$}
      \put(5,-4){$L(0)=\{a\}$}
      \node(b)(30,5){$1$}
      \put(25,-4){$L(1)=\{b\}$}
      \node(c)(50,5){$2$}
      \put(45,-4){$L(2)=\emptyset$}
      \drawedge(a,b){$a$}
      \drawloop[loopangle=90](a){$b$}
      \drawedge(b,c){$b$}
      \drawloop[loopangle=90](b){$a$}
      \drawloop[loopangle=90](c){$a$}
      \drawloop[loopangle=0](c){$b$}
    \end{picture}
    \caption{A finite-state observer Obs}
    \label{fig-mealy2}
\end{figure}
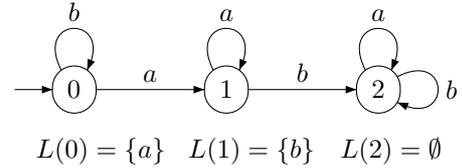
Let Obs be the observer of Figure~\ref{fig-mealy2}.  Obs maps the
following inputs as follows: $\text{Obs}(baab)=ab$,
$\text{Obs}(bababbaab)=ab$, $\text{Obs}(bbbbba)=a$ and
$\text{Obs}(bbaaa)=a$.  If $\text{Obs}$ operates on the DES $\calB$ of
Figure~\ref{fig-ex3} and $\calB$ generates $f.a.b$, $\text{Obs}$ will
have as input $\proj{\Sigma}(f.a.b)=a.b$ with
$\Sigma=\{a,b\}$. Consequently the observation of $\text{Obs}$ is
$\text{Obs}(\proj{\Sigma}(f.a.b))=a.b$.
\end{example}

\subsection{Fault Diagnosis with Dynamic Diagnosers}
\label{sec-obsk-diag}
\begin{definition}[$(\obs,k)$-diagnoser] \label{def-obsk-diag} Let $A$
  be a finite automaton over $\Sigma^{\varepsilon,f}$ and \obs be an
  observer over $\Sigma$. $D:\Sigma^* \rightarrow \{0,1\}$ is an
  \emph{$(\obs,k)$-diagnoser} for $A$ if 
  \begin{itemize}
  \item $\forall \rho \in \nonfaulty(A)$,
    $D(\obs(\proj{\Sigma}(\trace(\rho))))=0$ and 
  \item  $\forall \rho \in \faulty_{\geq k}(A)$,
    $D(\obs(\proj{\Sigma}(\trace(\rho))))=1$.     \endef
  \end{itemize}
\end{definition}
$A$ is $(\obs,k)$-diagnosable if there is an $(\obs,k)$-diagnoser for
$A$. $A$ is Obs-diagnosable if there is some $k$ such that $A$ is
$(\obs,k)$-diagnosable.

If a diagnoser always selects $\Sigma$ as the set of observable
events, it is a static observer and $(\obs,k)$-diagnosability amounts
to the standard $(\Sigma,k)$-diagnosis problem~\cite{Raja95}. 

As for $\Sigma$-diagnosability, we have the following equivalence for
dynamic observers: $A$ is $(\obs,k)$-diagnosable iff
\[\obs(\proj{\Sigma}(\faulty^{\textit{tr}}_{\geq k}(A))) \cap
\obs(\proj{\Sigma}(\nonfaulty^{\textit{tr}}(A))) = \emptyset \mathpunct.
\] 

\begin{prob}[Finite-State Obs-Diagnosability] \label{prob-mealy-diag} \mbox{} \\
  \textsc{Input:} $A$, \obs a finite-state observer.\\
  \textsc{Problem:} 
    \begin{enumerate}[(A)]
    \item Is $A$ \obs-diagnosable?
    \item If the answer to (A) is ``yes'', compute the minimum
      $k$ such that $A$ is $(\obs,k)$-diagnosable.
    \end{enumerate}
 \end{prob}  

 \begin{theorem}\label{thm-3}
   Problem~\ref{prob-mealy-diag} is in P.
 \end{theorem}
To prove Theorem~\ref{thm-3} we build 
a \emph{product} automaton\footnote{We use $\otimes$ to clearly
  distinguish this product from the usual synchronous product
  $\times$.} $A \otimes \obs$ such that: $A$ is $(\obs,k)$-diagnosable
$\iff$ $A \otimes \obs$ is $(\Sigma,k)$-diagnosable.
Given two finite automata
$A=(Q,q_0,\Sigma^{\varepsilon,f},\rightarrow)$ and
$\obs=(S,s_0,\Sigma,\delta,L)$,
the automaton $A \otimes \obs=(Q\times
S,(q_0,s_0),\Sigma^{\varepsilon,f},\rightarrow)$ is defined as
follows:
\begin{itemize}
\item $(q,s) \xrightarrow{\,\beta\,} (q',s')$ iff 
  $\exists \lambda \in \Sigma$ \st $q \xrightarrow{\,\lambda\,} q'$,
  $s'=\delta(s,\lambda)$ and $\beta=\lambda$ if $\lambda \in L(s)$,
  $\beta=\varepsilon$ otherwise;
\item $(q,s) \xrightarrow{\,\lambda\,} (q',s)$ iff 
  $\exists \lambda \in \{\varepsilon,f\}$ \st $q
  \xrightarrow{\,\lambda\,} q'$.
\end{itemize}
The number of states of $A \otimes \obs$ is at most $|Q| \times |S|$
and the number of transitions is bounded by the number of transitions
of $A$. Hence the size of the product is polynomial in the size of the
input $|A|+|\obs|$. Checking that $A \otimes \obs$ is diagnosable can
be done in polynomial time and 
Problem~\ref{prob-mealy-diag}.(A) is in P.

\begin{example}
  Let $\calB$ be the DES given in Figure~\ref{fig-ex3} and $\obs$ the
  observer of Figure~\ref{fig-mealy2}. The product $\calA \otimes \obs$
  \stavros{used in the above proof} is given in Figure~\ref{fig-prod}.
\end{example}
  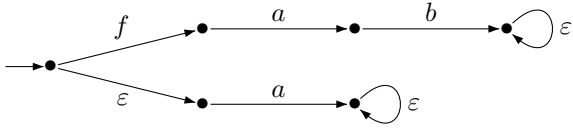
\begin{figure}[hbtp]
    \centering
    \begin{picture}(75,18)(4,0)
      \gasset{Nframe=n,Nw=2,Nh=2,loopdiam=5,loopangle=0}
      \node[Nmarks=i](a)(10,10){$\bullet$} \node(b)(30,15){$\bullet$}
      \node(c)(50,15){$\bullet$} \node(d)(70,15){$\bullet$}
      \node(dd)(50,5){$\bullet$} \node(e)(30,5){$\bullet$} 
      \drawloop(dd){$\varepsilon$}
      \drawloop(d){$\varepsilon$}
      \drawedge(a,b){$f$}
      \drawedge(b,c){$a$}
      \drawedge(c,d){$b$}
      \drawedge[ELside=r,curvedepth=0](a,e){$\varepsilon$}
      \drawedge[curvedepth=0,ELside=l](e,dd){$a$}
    \end{picture}
    \caption{The product $\calA \otimes \obs$
    }
    \label{fig-prod}
  \end{figure}

  For Problem~\ref{prob-mealy-diag}, we have assumed that an observer
  was given.  It would be even better if we could \emph{synthesize} an
  observer \obs such that the plant \stavros{is \obs-diagnosable}.  Before
  attempting to synthesize such an observer, we should first check
  that the plant is $\Sigma$-diagnosable: if it is not, then obviously
  no such observer exists; if the plant is $\Sigma$-diagnosable, then
  the trivial observer that observes all events in $\Sigma$ at all
  times works\footnote{ Notice that this also shows that existence of
    an observer implies existence of a finite-state observer, since
    the trivial observer is finite-state.  }. 
\stavros{
As a first step towards synthesizing non-trivial observers, we can
attempt to compute the set of {\em all} valid observers, which includes
the trivial one but also non-trivial ones (if they exist).
}

 \begin{prob}[Dynamic-Diagnosability] \label{prob-all-obs} \mbox{} \\
   \textsc{Input:} $A$.\\
   \textsc{Problem:} Compute the set $\calO$ of all observers such
   that $A$ is \obs-diagnosable iff $\obs \in \calO$.
 \end{prob}  
 We do not have a solution to the above general problem.
 Instead, we introduce a restricted variant:

 \begin{prob}[Dynamic-$k$-Diagnosability] \label{prob-all-kobs} \mbox{} \\
   \textsc{Input:} $A$, $k \in \setN$.\\
   \textsc{Problem:} Compute the set $\calO$ of all observers such
   that $A$ is $(\obs,k)$-diagnosable iff $\obs \in \calO$.
 \end{prob}  

\subsection{Problem~\ref{prob-all-kobs} as a Game Problem}
\label{sec-game}
To solve Problem~\ref{prob-all-kobs} we reduce it to a \emph{safety}
2-player game. 
%
In short, the reduction we propose is the following:
\begin{itemize}
\item Player~1 chooses the set of events it wishes
to observe, then it hands over to Player~2;
\item Player~2 chooses an event and tries to produce a run which is
  the observation of a $k$-faulty run and a non-faulty run.
\end{itemize}
Player~2 wins if he can produce such a run. Otherwise Player~1 wins.
Player~2 has complete information of Player~1's moves (i.e., it can
observe the sets that Player~1 chooses to observe).
Player~1, on the other hand, only has partial information of Player~2's
moves because not all events are observable (details follow).
Let $A=(Q,q_0,\Sigma^{\varepsilon,f},\rightarrow)$ be a finite automaton.
To define the game, we use two copies of automaton $A$: $A_1^k$ and
$A_2$. 
The accepting states of $A_1^k$ are those corresponding to runs of $A$
which are faulty and where more than $k$ steps occurred after the
fault.  $A_2$ is a copy of $A$ where the $f$-transitions have been
removed.  The game we are going to play is the following (see
Figure~\ref{fig-game}, Player~1 states are depicted with square boxes
and Player~2 states with round shapes):
 \begin{enumerate}
 \item the game starts in an state $(q_1,q_2)$ corresponding to the
   initial state of the product of $A_1^k$ and $A_2$. 
   Initially, it is Player~1's turn to play. Player~1 chooses a set of
   events he is going to observe \ie a subset $X$ of $\Sigma$ and
   hands it over to Player~2;
\item assume the automata $A_1^k$ and $A_2$ are in states $(q_1,q_2)$.
  Player~2 can change the state of $A_1^k$ and $A_2$ by:
  \begin{enumerate}
  \item firing an action (like $\lambda_1, \lambda_2, \lambda_3,
    \lambda_4$ in Figure~\ref{fig-game}) which is not in $X$ in either
    $A_1^k$ or $A_2$ (no synchronization).  In this case a new state
    $(q,q')$ is reached and Player~2 can play again from this state;
  \item firing an action in $X$ (like $\sigma_1,\sigma_2$
    in Figure~\ref{fig-game}): to do this both $A_1^k$ and $A_2$ must be
    in a state where $\lambda$ is possible (synchronization); after
    the action is fired a new state $(q_1',q_2')$ is reached: now it
    is Player~1's turn to play, and the game continues as in step~1
    above from the new state $(q_1',q_2')$.
  \end{enumerate}
 \end{enumerate}
 Player~2 wins if he can reach a state $(q_1,q_2)$ in $A_1^k
  \times A_2$ where $q_1$ is an accepting state of $A_1^k$ (this means
  that Player~1 wins if it can avoid ad infinitum this set of states).
  In this sense this is a safety game for Player~1 (and a reachability
  game for Player~2).
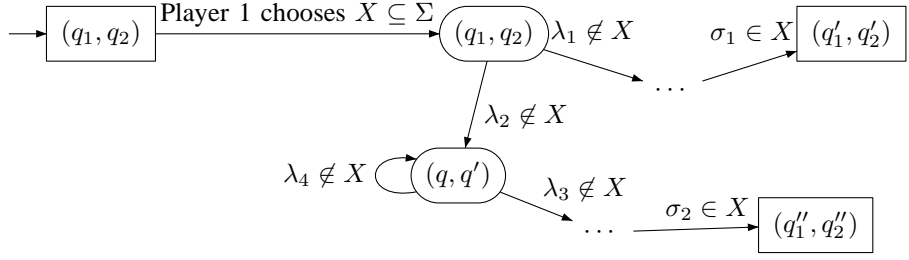
\begin{figure*}[hbtp]
    \centering
    \unitlength=.94mm
    \begin{picture}(125,43)(-22,-24)
      \gasset{Nframe=y,Nadjust=wh,Nadjustdist=2,loopdiam=5,loopangle=0}
      \node[Nmarks=i,Nmr=0,Nadjust=wh](a)(-10,10){$(q_1,q_2)$} 
      \node(b)(45,10){$(q_1,q_2)$}
      \node[Nframe=n](c)(70,2){$\cdots$}
      \node(c1)(40,-10){$(q,q')$}
      \node[Nframe=n](c2)(60,-18){$\cdots$}
      \node[Nmr=0,Nadjust=wh](d)(95,10){$(q'_1,q'_2)$} 
      \node[Nmr=0,Nadjust=wh](e)(90,-17){$(q''_1,q''_2)$} 
      \drawedge(a,b){Player~1 chooses $X \subseteq \Sigma$}
      \drawedge(b,c){$\lambda_1 \not\in X$}
      \drawedge[ELpos=50](c,d){$\sigma_1 \in X$}
      \drawedge(b,c1){$\lambda_2 \not\in X$}
      \drawedge[ELpos=80](c1,c2){$\lambda_3 \not\in X$}
      \drawedge(c2,e){$\sigma_2 \in X$}

      \drawloop[loopangle=180](c1){$\lambda_4 \not\in X$}
       \end{picture}
    \caption{Game reduction for problem~\ref{prob-all-kobs}}
    \label{fig-game}
\end{figure*}
Formally, the game $G_A=(S_1\uplus S_2,s_0,\Sigma_1 \uplus
\Sigma_2,\delta)$ is defined as follows ($\uplus$ denotes union of
disjoint sets):
\begin{itemize}
\item $S_1= (Q\times \{-1,\cdots,k\}) \times Q$ is the set of Player~1
  states; a state $((q_1,j),q_2) \in S_1$ indicates that $A_1^k$ is in
  state $q_1$, $j$ steps have occurred after a fault, and $q_2$ is the
  current state of $A_2$. If no fault has occurred, $j=-1$ and if more
  than $k$ steps occurred after the fault, we use $j=k$.
\item $S_2= (Q\times \{-1,\cdots,k\}) \times Q \times 2^{\Sigma}$ is
  the set of Player~2 states. For a state $((q_1,j),q_2,X) \in S_2$,
  the triple $((q_1,j),q_2)$ has the same meaning as for $S_1$, and $X$
  is the set of moves Player~1 has chosen to observe on its last move.
\item $s_0=((q_0,-1),q_0)$ is the initial state of the game belonging
  to Player~1;
\item $\Sigma_1=2^{\Sigma}$ is the set of moves of Player~1;
  $\Sigma_2=\Sigma^\varepsilon$ is the set of moves of Player~2 (as we
  encode the fault into the state, we do not need to distinguish $f$
  from $\varepsilon$).
\item the transition relation $\delta \subseteq (S_1 \times \Sigma_1
  \times S_2) \cup (S_2 \times \{\varepsilon\} \times S_2) \cup (S_2
  \times \Sigma \times S_1)$ is defined by:
  \begin{itemize}
  \item Player~1 moves: let $\sigma \in \Sigma_1$
    and $s_1 \in S_1$. Then $(s_1, \sigma, (s_1,\sigma)) \in \delta$.
  \item Player~2 moves: a move of Player~2 is either a silent move
    ($\varepsilon$) \ie a move of $A_1^k$ or $A_2$ or a joint move of
    $A_1^k$ and $A_2$ with an observable action in $X$.  Consequently,
    a \emph{silent} move
    $((q_1,i),q_2,X),\varepsilon,(q'_1,j),q'_2,X)) $ is in $\delta$ if
    one of the following conditions holds:
    \begin{enumerate}
    \item either $q'_2=q_2$, $q_1 \xrightarrow{\,\ell\,} q'_1$ is a
      step of $A_1^k$, $\ell \not\in X$, and if $i \geq 0$ then
      $j=\min(i+1,k)$; if $i=-1$ and $\ell=f$ $j=0$ otherwise $j=i$.
    \item either $q'_1=q_1$, $q_2 \xrightarrow{\,\ell\,} q'_2$ is a
      step of $A_2$, $\ell \not\in X$ (and $\ell \neq f$), and if $i
      \geq 0$ then $j=\min(i+1,k)$, otherwise $j=i$.
    \end{enumerate}
    A \emph{visible} move can be taken by Player~2 if both $A_1^k$ and
    $A_2$ agree on doing such a move. In this case the game proceeds
    to a Player~1 state: $((q_1,i),q_2,X),\ell,((q'_1,j),q'_2)) \in
    \delta$ if $\ell \in X$, $q_1 \xrightarrow{\,\ell\,} q'_1$ is a
    step of $A_1^k$, $q_2 \xrightarrow{\,\ell\,} q'_2$ is a step of
    $A_2$, and if $i \geq 0$ then $j=\min(i+1,k)$, otherwise $j=i$.
  \end{itemize}
\end{itemize}
We can show that for any observer $O$ \st $A$ is
$(O,k)$-dia\-gnosable, there is a strategy $f(O)$ for Player~1 in
$G_A$ \st $f(O)$ is \emph{trace-based} and winning.  A \emph{strategy}
for Player~1 is a mapping $f: \runs(G_A) \rightarrow \Sigma_1$ that
associates a move $f(\rho)$ in $\Sigma_1$ to each run $\rho$ in $G_A$
that ends in an $S_1$-state. A strategy $f$ is trace-based if given
two runs $\rho,\rho'$, if $\trace(\rho)=\trace(\rho')$ then
$f(\rho)=f(\rho')$.
Conversely, for any trace-based winning strategy $f$ (for Player~1),
we can build an observer $O(f)$ \st $A$ is $(O(f),k)$-diagnosable.

Let $O=(S,s_0,\Sigma,\delta,L)$ be an observer for
$A$. 
We define the strategy $f(O)$ on finite runs of $G_A$ ending in a
Player~1 state by: $f(O)(\rho)
=L(\delta(s_0,\proj{\Sigma}(\trace(\rho))))$.
The intuition is that we take the run $\rho$ in $G_A$, take the trace
of $\rho$ (choices of Player~1 and moves of Player~2) and remove the
choices of Player~1. This gives a word in $\Sigma^*$. The strategy for
Player~1 for $\rho$ is the set of events the observer $O$ chooses to
observe after reading $\proj{\Sigma}(\trace(\rho))$ \ie
$L(\delta(s_0,\proj{\Sigma}(\trace(\rho))))$.

\noindent Conversely, with each trace-based strategy $f$ of the game
$G_A$ we can associate an automaton $O(f)=(S,s_0,\Sigma,
\delta,L)$ defined by:
\begin{itemize}
\item $S= \{\proj{\Sigma}(\trace(\rho)) \,|\, \rho \in \out{G_A,f}
  \textit{ and } \last(\rho) \in S_1 \}$;
\item $s_0=\varepsilon$;
\item $\delta(v,\ell)= v'$ if $v \in S$, $v'=v.\ell$ and there is a
  run $\rho \in \out{G_A,f}$ with $\rho=q_0 \xrightarrow{X_0} q_0^1
  \xrightarrow{\varepsilon^*} q_0^{n_0} \xrightarrow{\lambda_1} q_1
  \xrightarrow{X_1} q_1^1 \xrightarrow{\varepsilon^*} q_1^{n_1}
  \xrightarrow{\lambda_2} q_2 \cdots q_{k_1}
  \xrightarrow{\varepsilon^*} q_{k-1}^{n_{k-1}}
  \xrightarrow{\lambda_k} q_k$ with each $q_i \in S_1$, $q_i^j \in
  S_2$, $v=\proj{\Sigma}(\trace(\rho))$, and $\rho
  \xrightarrow{X_k} q_{k}^1 \xrightarrow{\varepsilon^*} q_{k}^{n_{k}}
  \xrightarrow{\ell} q_{k+1}$ with $q_{k+1} \in S_1$, $\ell \in X_k$.
  \\
  $\delta(v,l)=v$ if $v \in S$ and $\ell \not\in f(\rho)$;
\item $L(v)=f(\rho)$ if $v=\proj{\Sigma}(\trace(\rho))$.
\end{itemize}
Using the previous definitions and constructions we obtain the
following theorems:
\begin{theorem}\label{red-1}
  Let $O$ be an observer \st $A$ is $(O,k)$-dia\-gno\-sa\-ble. Then
  $f(O)$ is a trace-based winning strategy in $G_A$.
\end{theorem}

\begin{theorem}\label{red-2}
  Let $f$ be a trace-based winning strategy in $G_A$. Then $O(f)$ is
  an observer and $A$ is $(O(f),k)$-diagnosable.
\end{theorem}

The result on a game like $G_A$ is that, if there is a winning
trace-based strategy for Player~1, then there is a most permissive
strategy $\calF_A$ which has finite memory.  It can be represented by
a finite automaton $S_{\calF_A}=(W_1 \uplus W_2,s_0,\Sigma \cup
2^\Sigma,\Delta_A)$ \st $\Delta_A \subseteq (W_1 \times 2^\Sigma
\times W_2) \cup (W_2 \times \Sigma \times W_1)$ which has size
exponential in the size of $G_A$.  For a given run $\rho \in (\Sigma
\cup 2^{\Sigma})^*$ ending in a $W_1$-state, we have
$\calF_A(w)=\enabled(\Delta_A(s_0,w))$.

\subsection{Most Permissive Observer}
\label{sec-most-permissive}

We now define the notion of a most \emph{permissive} observer and show
the existence of a most permissive observer for a system in case $A$
is diagnosable.  $\calF_A$ is the mapping defined at the end of the
previous section.

For an observer $O=(S,s_0,\Sigma,\delta,L)$ and $w \in \Sigma^*$ we
let $L(w)$ be the set $L(\delta(s_0,w))$: this is the set of events
$O$ chooses to observe on input $w$.  Given a word $\rho \in
\proj{\Sigma}(\lang(A))$, we
recall that $O(\rho)$ is the observation of $\rho$ by $O$. Assume
$O(\rho)=a_0 \cdots a_k$.  Let
$\overline{\rho}=L(\varepsilon).\varepsilon.L(a_0).a_0.\cdots
L(O(\rho)(k)).a_k$ \ie $\overline{\rho}$ contains the history of what
$O$ has chosen to observe at each step and the events that occurred
after each choice.

Let $\calO: (2^\Sigma \times \Sigma^\varepsilon )^+ \rightarrow
2^{2^\Sigma}$. By definition $\calO$ is the most permissive observer
for $(A,k)$ if the following holds:
\begin{equation*}
   \label{eq-most-perm-obs}
  \begin{array}{c}
   \textit{$O=(S,s_0,\Sigma,\delta,L)$ } \\
    \textit{is an observer and} \\
    \textit{and $A$ is $(O,k)$-diagnosable}
  \end{array}
   \iff 
   \begin{array}{c}
     \textit{$\forall w \in \Sigma^*$, } \\
     \textit{$L(\delta(s_0,w)) \in \calO(\overline{w})$}
   \end{array}
\end{equation*}
The definition of the most permissive observer states that:
\begin{itemize}
\item any good observer $O$ (one such that $A$ is
  $(O,k)$-diagnosable) must choose a set of observable events
  in $\calO(\overline{w})$ on input $w$;
\item if an observer chooses its set of observable events in
  $\calO(\overline{w})$ on input $w$, then it is a good observer.
\end{itemize}

Assume $A$ is $(\Sigma,k)$-diagnosable. Then there is an observer $O$
\st  $A$ is $(O,k)$-diagnosable because the constant observer that
observes $\Sigma$ is a solution.  By Theorem~\ref{red-1}, there is a
trace-based winning strategy for Player~1 in $G_A$. 

\begin{theorem}
  \label{thm-most-perm}
  $\calF_A$ is
the most permissive observer.
\end{theorem}
%
%
This enables us to solve Problem~\ref{prob-all-kobs} and compute a
finite representation of the set $\calO$ of all observers such that
$A$ is $(O,k)$-diagnosable iff $O \in \calO$.
Computing $\calF_A$ can be done in $O(2^{|G_A|})$. The size of $G_A$
is quadratic in $|A|$, linear in the size of $k$, and exponential in
the size of $\Sigma$ \ie $|G_A|=O(|A|^2 \times 2^{|\Sigma|} \times
|k|)$.  This means that computing $\calF_A$ can be done in exponential
time in the size of $A$ and $k$ and doubly exponential time in the
size of $\Sigma$.

The computation of a \emph{generic} diagnoser associated with the most
permissive observer can de done as well. This diagnoser is the
\emph{most permissive dynamic} diagnoser and contains all the choices
a dynamic diagnoser can make to be able to diagnose a plant.

\section{Optimal Dynamic Observers}
\label{sec-opt-pb}

In this section we define a notion of cost for observers.  This will
allow us to compare observers w.r.t. to this criterion and later on to
synthesize an optimal observer.  The notion of cost we are going to
use is inspired by \emph{weighted automata}.
  
\subsection{Weighted Automata \& Karp's Algorithm}
The notion of cost for automata has already been defined and
algorithms to compute some optimal values related to this model are
described in many papers. We recall here the results of~\cite{karp-78}
which will be used later.
\begin{definition}[Weighted Automaton] 
  A \emph{wei\-ghted au\-to\-ma\-ton} is a pair $(A,w)$ s.t.
  $A=(Q,q_0,\Sigma,\delta)$ is a finite automaton and $w: Q
  \rightarrow \setN$ associates a weight with each state.  \endef
\end{definition}
\begin{definition}[Mean Cost]
  Let $\rho=q_0\xrightarrow{a_2}q_1\xrightarrow{a_1} \cdots
  \xrightarrow{a_{n}}q_n$ be a run of $A$. The \emph{mean cost} of
  $\rho$ is 
\[
\mu(\rho)=\frac{1}{n+1} \times \sum_{i=0}^n
  w(q_i) \, \mathpunct. 
\] \endef 
\end{definition}
We remind that the length of
$\rho=q_0\xrightarrow{a_1}q_1\xrightarrow{a_2} \cdots
\xrightarrow{a_{n}}q_n $ is $|\rho|=n$.  We assume that $A$ is
complete w.r.t. $\Sigma$ (and $\Sigma \neq \emptyset$) and thus
contains at least one run for any arbitrary length $n$.  Let
$\runs^n(A)$ be the set of runs of length $n$ in $\runs(A)$.
The \emph{maximum mean-weight} of the runs of length $n$ for $A$ is
$\nu(A,n)=\max \{ \mu(\rho) \textit{ for } \rho \in \runs^n(A) \}$.
The \emph{maximum mean weight} of $A$ is $\nu(A)=\lim \sup_{n
  \rightarrow \infty}\nu(A,n)$.  Actually the value $\nu(A)$ can be
computed using Karp's maximum mean-weight cycle
algorithm~\cite{karp-78} on weighted graphs.  If $c=s_0
\xrightarrow{a_1}s_1\xrightarrow{a_2} \cdots \xrightarrow{a_{n}}s_n$
is a cycle of $A$ \ie $s_0=s_n$, the \emph{mean weight} of the cycle
$c$ is $\mu(c)=\frac{1}{n+1} \cdot \sum_{i=0}^n w(s_i)$.  The
\emph{maximum mean-weight cycle} of $A$ is the value $\nu^*(A) = \max
\{ \mu(c) \textit{ for $c$ a cycle of $A$} \}$.  As stated
in~\cite{zwick-95}, for weighted automata, the mean-weight cycle value
is the value that determines the mean-weight value: $\nu(A)=\lim
\sup_{n \rightarrow \infty}\nu(A,n)=\lim_{n \rightarrow
  \infty}\nu(A,n)=\nu^*(A)$.

\subsection{Cost of a Dynamic Observer}

\label{sec-cost-obs}
Let $\obs=(S,s_0,\Sigma,\delta,L)$ be an observer and
$A=(Q,$ $q_0,\Sigma^{\varepsilon,f},\rightarrow)$.
We would like to define a notion of \emph{cost} for observers in order
to select an optimal one among all of those which are valid, \ie
s.t. $A$ is $(\obs,k)$-diagnosable.  Intuitively this notion of cost
should imply that the more events we observe at each time,
the more expensive it is.


There is not one way of defining a notion of cost and the reader is
referred to~\cite{cassez-tase-07} for a discussion on this subject.

The cost of a word $w$ is given by:
\[ \cost(w)=\frac{\sum_{i=0}^{i=n}|L(\delta(s_0,w(i)))|}{n+1}\] with
$n=|w|$.  
We now show how to define and compute the cost of an observer $\obs$
that observes a DES $A$.

Given a run $\rho \in \runs(A)$, the observer only processes
$\proj{\Sigma}(\trace(\rho))$ ($\varepsilon$ and $f$-transitions are
not processed). To have a consistent notion of costs that takes into
account the logical time elapsed from the beginning, we need to take
into account one way or another the number of \emph{steps} of $\rho$
(the length of $\rho$) even if some of them are non observable.  A
simple way to do this is to consider that $\varepsilon$ and $f$ are
now observable events, let's say $u$, but that the observer never
chooses to observe them. Indeed we assume we have already checked that
$A$ is $(\obs,k)$-diagnosable, and the problem is now to compute the
cost of the observer we have used.

\begin{definition}[Cost of a Run]\label{def-cost-run}
  Given a run $\rho=q_0 \xrightarrow{\ a_1\ } q_1 \xrightarrow{\ a_2\
  } \cdots q_{n-1} \xrightarrow{\ a_n\ } q_n \in \runs(A)$, let $w_i =
  \obs(\proj{\Sigma}(\trace(\rho(i)))), 0 \leq i \leq n$.  The
  \emph{cost} of $\rho \in \runs(A)$ is defined by:
\[
\cost(\rho,A,\obs)=\frac{1}{n+1} \cdot \sum_{i=0}^n |L(\delta(s_0,w_i)|
\mathpunct.
\] \endef
\end{definition}
We recall that $\runs^n(A)$ is the set of runs of length $n$ in
$\runs(A)$.  The cost of the runs of length $n$ of $A$ is
\[
\cost(n,A,\obs)=\max \{ \cost(\rho,A,\obs) \textit{ for } \rho
\in \runs^n(A) \} \mathpunct .
\]
  The cost of the pair $(\obs,A)$ is
\[
\cost(A,\obs)=\limsup_{n \rightarrow \infty} \cost(n,A,\rho) \mathpunct .
\]
Notice that $\cost(n,A,\obs)$ is defined for each $n$ because we
have assumed $A$ generates runs of arbitrary large length.

As emphasised previously, in order to compute $\cost(n,A,\obs)$ we
consider that $\varepsilon$ and $f$ are now observable events, say
$u$, but that the observer never chooses to observe them.
Let $\obs^+=(S,s_0,\Sigma^u,\delta',L)$ where $\delta'$ is $\delta$
augmented with $u$-transitions that loop on each state $s \in S$.  Let
$A^+$ be $A$ where $\varepsilon$ and $f$ transitions are renamed $u$.
Let $A^+ \times \obs^+$ be the synchronized product of $A^+$ and
$\obs^+$.  $A^+ \times \obs^+=(Z,z_0,\Sigma^u,\Delta)$ is complete
w.r.t. $\Sigma^u$ and we let $w(q,s)=|L(s)|$ so that $(A^+ \times
\obs^+,w)$ is a weighted automaton.

\begin{theorem}
  $\cost(A,\obs)=\nu^*(A^+ \times \obs^+)$.
\end{theorem}
\noindent Thus we can compute the cost of a given pair $(A,\obs)$:
this can be done using Karp's maximum mean weight cycle
algorithm~\cite{karp-78} on weighted graphs.  This algorithm is
polynomial in the size of the weighted graph and thus:
\begin{theorem}
  Computing the cost of $(A,\obs)$ is in P.
\end{theorem}

\begin{remark}
  Notice that instead of the values $|L(s)|$ we could use any mapping
  from states of Obs to $\setZ$ and consider that the cost of
  observing $\{a,b\}$ is less than observing $a$.
\end{remark}

\subsection{Optimal Dynamic Diagnosers}
\label{sec-optimal}
In this section, we focus on the problem of computing a best observer
in the sense that diagnosing the DES with it has minimal cost.  We
address the following problem:
\begin{prob}[Bounded Cost Observer] \label{prob-bounded-cost} \mbox{} \\
  \textsc{Input:} $A$, $k \in \setN$ and $c \in \setN$.\\
  \textsc{Problem:}
    \begin{enumerate}[(A).]
    \item Is there an observer \obs s.t. $A$ is
      (\obs,k)-diagnosable and $\cost(\obs) \leq c$ ?
    \item If the answer to (A) is ``yes'', compute a witness optimal
      observer \obs with $\cost(\obs) \leq c$.
    \end{enumerate}
 \end{prob}  

%
 Theorem~\ref{thm-most-perm}, page~~\pageref{thm-most-perm}
 establishes that there is a most permissive observer $\calF_A$ in
 case $A$ is $(\Sigma,k)$-diagnosable and it can be computed in
 exponential time in the size of $A$ and $k$, doubly exponential time
 in $|\Sigma|$, and has size exponential in $A$ and $k$, and doubly
 exponential in $|\Sigma|$.  Moreover the most permissive observer
 $\calF_A$ can be represented by a finite state machine
 $S_{\calF_A}=(\{0,2\cdots,l\} \cup (\{1,3,\cdots,2l'+1\}\times
 2^\Sigma),0,\Sigma \cup 2^\Sigma,\delta)$ which has the following
 properties:
\begin{itemize}
\item even states are states where the observer chooses a set of events
  to observe;
\item odd states $(2i+1,X)$ are states where the observer waits for 
  an observable event in $X$ to occur;
\item if $\delta(2i,X)=(2i'+1,X)$ with $X \in 2^\Sigma$, it means that
  from an even state $2i$, the automaton $S_{\calF_A}$ can select a
  set $X$ of events to observe. The successor state is an odd state
  together with the set $X$ of events that are being observed;
\item if $\delta((2i+1,X),a)=2i'$ with $a \in X$, it means that from
  $(2i+1,X)$, $S_{\calF_A}$ is waiting for an observable event to
  occur. When some occurs it switches to an even state.
\end{itemize}
By definition of $\calF_A$, any observer $O$ s.t. $A$ is
$(O,k)$-diagnosable must select a set of observable events in
$\calF_A(\trace(\overline{w}))$ after having observed $w \in
\proj{\Sigma}(\lang(A))$.

To compute an optimal observer, we use a result by Zwick and
Paterson~\cite{zwick-95} on \emph{weighted graph games}. 

%
%
\begin{definition}[Weighted Graph] \label{def-w-g} A \emph{wei\-ghted
    directed graph} is a pair $(G,w)$ s.t.  $G=(V,E)$ is a directed
  graph and $w: E \rightarrow \{-W,\cdots,0,\cdots,W\}$ assigns an
  integral weight to each edge of $G$ with $W \in \setN$.  We assume
  that each vertex $v \in V$ is reachable from a unique \emph{source}
  vertex $v_0$ and has at least one outgoing transition.
 \endef
\end{definition}

\begin{definition}[Weighted Graph Game]
  A \emph{weighted graph game} $G=(V,E)$ is a bipartite weighted graph
  with $V=V_1 \cup V_2$ and $E=E_1 \cup E_2$, $E_1 \subseteq V_1
  \times V_2$ and $E_2 \subseteq E_2 \times E_1$. We assume the
  initial vertex $v_0$ of $G$ belongs to $V_1$. \endef
\end{definition}
Vertices $V_i$ are Player~i's vertex. A weighted graph game is a turn
based game in which the turn alternates between Player~1 and Player~2.
The game starts at a vertex $v_0 \in V_1$. Player~1 chooses an edge
$e_1=(v_0,v_1)$ and then Player~2 chooses an edge $e_2=(v_1,v_2)$ and
so on and they build an infinite sequence of edges.
Player~1 wants to maximise $\liminf_{n \rightarrow \infty} \frac{1}{n}
\cdot \sum_{i=1}^{n}w(e_i)$ and Player~2 wants to minimize $\limsup_{n
  \rightarrow \infty} \frac{1}{n} \cdot \sum_{i=1}^{n}w(e_i)$.

One of the result of~\cite{zwick-95} is that there is a rational value
$\nu \in \setQ$ s.t. Player~1 has a strategy to ensure $\liminf_{n
  \rightarrow \infty} \frac{1}{n} \cdot \sum_{i=1}^{n}w(e_i) \geq \nu$
and Player~2 has a strategy to ensure that $\limsup_{n \rightarrow
  \infty} \frac{1}{n} \cdot \sum_{i=1}^{n}w(e_i) \leq \nu$. $\nu$ is
called the value of the game. 
In summary the results by Zwick and Paterson~\cite{zwick-95} we are
going to use are:
 \begin{itemize}
\item there is a value $\nu \in \setQ$, called the \emph{value of the
    game} s.t. 
  Player~1 has a strategy to ensure that $\liminf_{n \rightarrow
    \infty} \frac{1}{n}\sum_{i=1}^n w(e_i) \geq \nu$ and Player~2 has
  a strategy to ensure that $\limsup_{n \rightarrow \infty}
  \frac{1}{n}\sum_{i=1}^n w(e_i) \leq \nu$; this value can be computed
  in $O(|V|^3 \times |E| \times W)$ where $W$ is the range of the
  weight function (assuming the weights are in the interval
  $[-W..W]$). Note that deciding whether this value satisfies $\nu
  \bowtie c$ for $\bowtie \in \{=,<,>\}$ for $c \in \setQ$ can be done
  in $O(|V|^2 \times |E| \times W)$.
\item there are optimal memoryless strategies for both players that
  can be computed in $O(|V|^4 \times |E| \times \log(|E|/|V|) \times
  W)$.
\end{itemize}

To solve Problem~\ref{prob-bounded-cost}, we use the most permissive
observer $\calF_A$ we computed in section~\ref{sec-most-permissive}.
Given $A$ and $\calF_A$, we build a weighted graph game $G(A,\calF_A)$
s.t.  the value of the game is the optimal cost for the set of all
observers.  Moreover an optimal observer can be obtained by taking an
optimal memoryless strategy in $G(A,\calF_A)$.

To build $G(A,\calF_A)$ we use the same idea as in
section~\ref{sec-cost-obs}: we replace $\varepsilon$ and $f$
transitions in $A$ by $u$ obtaining $A^+$.  We also modify $\calF_A$
to obtain a weighted graph game $(\calF_A^+,w)$ by adding transitions
so that each state $2k+1$ is complete w.r.t. $\Sigma^u$. This is done
as follows:
\begin{itemize}
\item from each $(2i+1,X)$ state, create a new even state \ie pick
  some $2i'$ that has not already been used. Add transitions
  $((2i+1,X),\sigma,2i')$ for each $\sigma \in \Sigma^u \setminus
  \enabled(2i+1,X)$. Add also a transition $(2i',X,(2i+1,X))$. This
  step means that if a $A$ produces an event and it is not observable,
  $\calF_A^+$ just reads the event and makes the same choice again.
\item the weight of a transition $(2i,X,(2i'+1,X))$ is $|X|$.
\end{itemize}
The game $G(A,\calF_A)$ is then $A^+ \times \calF_A^+$.  
This way we can obtain a weighted graph game $WG(A,\calF_A)$ by
abstracting away the labels of the transitions.  Notice that it still
enables us to convert any strategy in $WG(A,\calF_A)$ to a strategy in
$\calF_A$. A strategy in $WG(A,\calF_A)$ will define an edge
$(2i,(2i'+1,X))$ to take. As the target vertex contains the set of
events we chose to observe we can define a corresponding strategy in
$\calF_A$.

By construction of $G(A,\calF_A)$ and the definition of the value of a
weighted graph game, the value of the game is the optimal cost for the
set of all observers $O$ s.t. $A$ is $(O,k)$-diagnosable.  

Assume $A$ has $n$ states and $m$ transitions.  From
Theorem~\ref{thm-most-perm} we know that $\calF_A$ has at most
$O(2^{n^2} \times 2^k \times 2^{2^{|\Sigma|}})$ states and $O(2^{n^2}
\times 2^k \times 2^{2^{|\Sigma|}} \times n^2 \times k \times m)$
transitions.  Hence $G(A,\calF_A)$ has at most $O( n \times 2^{n^2}
\times 2^k \times 2^{2^{|\Sigma|}})$ vertices and $O(m \times 2^{n^2}
\times 2^k \times 2^{2^{|\Sigma|}})$ edges.  To make the game complete
we may add at most half the number of states and hence $WG(A,\calF_A)$
has the same size.  We thus obtain the following results:
\begin{theorem}
  Problem~\ref{prob-bounded-cost} can be solved in time $O(|\Sigma|
  \times m \times 2^{n^2} \times 2^k \times 2^{2^{|\Sigma|}})$.
\end{theorem}
We can even solve the optimal cost computation problem:
\begin{prob}[Optimal Cost Observer] \label{prob-optimal-cost} \mbox{} \\
  \textsc{Input:} $A$, $k \in \setN$.\\
  \textsc{Problem:} Compute the least value $m$ s.t. there exists an
  observer \obs s.t. $A$ is (\obs,k)-diagnosable and $\cost(\obs)
  \leq m$.
 \end{prob}  
\begin{theorem}\label{thm-optimal-cost}
  Problem~\ref{prob-optimal-cost} can be solved in time $O(|\Sigma|
  \times m \times 2^{n^2} \times 2^k \times 2^{2^{|\Sigma|}})$.
 \end{theorem}
 A consequence of Theorem~\ref{thm-optimal-cost} and Zwick and
 Paterson's results is that the cost of the optimal observer is a
 rational number.

\section{Conclusions}
\label{sec_ccl}

In this paper we have addressed sensor minimization problems in the
context of fault diagnosis, using dynamic
observers.
We proved that, for an observer given by a finite automaton,
diagnosability can be checked in polynomial time (as in the case of
static observers).  We also solved \stavros{a} synthesis problem of
dynamic observers and showed that a most-permissive dynamic observer
can be computed in doubly-exponential time, \stavros{provided an upper
  bound on the delay needed to detect a fault is given}.  Finally we
have \stavros{ defined a notion of cost for dynamic obervers and shown
  how to compute the minimal-cost observer that can be used to detect
  faults within a given delay.  }


There are several directions we are currently investigating.

Problem~\ref{prob-all-obs} has not been solved so far.  The major
impediment to solve it is that the reduction we propose in
section~\ref{sec-dyn-obs} yields a Büchi game in this case. 
More generally we \stavros{plan to} extend the framework we have
introduced for fault diagnosis to control under dynamic partial
observation and this will enable \stavros{us} to solve
Problem~\ref{prob-all-obs}.

Problem~\ref{prob-all-kobs} is solved in doubly exponential
time. Nevertheless to reduce the number of states of the most
permissive observer, we point out that only \emph{minimal} sets of
events \stavros{ need to be observed}.  Indeed, if we can diagnose a
system by observing only $\Sigma$ from some point on, we surely can
diagnose it using any superset $\Sigma' \supseteq \Sigma$. So far we
keep all the sets that can be used to diagnose the system. We could
possibly take advantage of the previous property using techniques
described in~\cite{doyen-csl-06}.

\bibliographystyle{IEEEtran} 
\bibliography{../biblio,../diagnosis}

\end{document}